\documentclass[11pt,reqno]{amsart}

\usepackage{amscd,amssymb,amsmath,amsthm}
\usepackage[arrow,matrix]{xy}
\usepackage{graphicx}
\usepackage{cite}
\newtheorem{thm}{Theorem}
\newtheorem{defn}{Definition}
\newtheorem{lemma}{Lemma}
\newtheorem{pro}{Proposition}
\newtheorem{rk}{Remark}

\numberwithin{equation}{section} 

\newcommand{\bea}{\begin{eqnarray}}
\newcommand{\eea}{\end{eqnarray}}

\newcommand{\R}{\mathbb{R}}

\def \t {\theta}

\begin{document}
\title[On four state HC models]{On four state Hard Core Models on the Cayley Tree}

\author{D. Gandolfo, U. A. Rozikov, J. Ruiz}

 \address{D.\ Gandolfo and J.Ruiz\\Centre de Physique Th\'eorique, UMR 7332, 
 Aix Marseille Univ, Universit\'e de Toulon, CNRS, CPT, Marseille, France.}
\email {gandolfo@cpt.univ-mrs.fr\ \ ruiz@cpt.univ-mrs.fr}

 \address{U.\ A.\ Rozikov\\ Institute of mathematics,
29, Do'rmon Yo'li str., 100125, Tashkent, Uzbekistan.}
\email {rozikovu@yandex.ru}

\begin{abstract}
We consider a nearest-neighbor four state hard-core (HC) model on the homogeneous Cayley tree of order $k$.
The Hamiltonian of the model is considered on a set of ``admissible'' configurations.
Admissibility is specified through a graph with four vertices.
We first exhibit conditions (on the graph and on the parameters) under which the model has a unique Gibbs measure.
Next we turn to some specific cases.
Namely, first we study, in the case of a particular  graph (the diamond), translation-invariant and periodic Gibbs measures.
We provide in both cases the equations of the transition lines separating uniqueness  from non--uniqueness regimes.
Finally the  same is done for ``fertile'' graphs, the so--called stick, gun, and key (here only translation invariant states
are taken into account).
\end{abstract}
\maketitle



\section{Introduction and  definitions} \label{sec:intro}

Hard core constraints arise in fields as diverse as combinatorics,
statistical mechanics and telecommunications.
In particular, hard core
models arise in the study of random independent sets of graphs \cite{Br3}, \cite{Gal}, the
study of gas molecules on a lattice \cite{Ba},  in the analysis of multi-casting
in telecommunication networks (see e.g. \cite{Kel1}, \cite{Lo}, \cite{Mi}).

We refer the reader to the nice article by Brightwell and Winkler  \cite{Br1}
on the subject, and to \cite{Br3} focusing on hard core models on the Bethe lattice (Cayley tree).

 Let $\Gamma^k= (V , L)$ be the uniform Cayley tree,
 where each vertex has $k + 1$ neighbors with $V$ being the set of vertices and $L$ the set of edges (bonds).

On the Cayley tree, there is a natural distance to be denoted $d(x,y)$,
 being
 the number of nearest neighbors pairs  of the minimal path between  the vertices $x$ and $y$
 (by  path one means   a collection of  nearest neighbors pairs, two consecutive pairs
 sharing at least a given vertex).

For a fixed $x^0\in V$, the root,
we let
$$ V_n=\{x\in V\ \ | \ \  d(x,x^0)\leq n\}$$
be the ball of radius $n$
and
$$ W_n=\{x\in V\ \ |\ \  d(x^0,x)=n\}$$
be the sphere of radius $n$ with center at $x^0$.

We will write $x<y$ if the path from $x^0$ to $y$ goes through $x$.
This is a partial order on the tree, for example, two different points of $W_n$ can not be ordered by this way.
But for any pair of
nearest neighbors $x$ and $y$ one has $x<y$ or $x>y$.

Let $S(x)$ be the direct successors of $x$, i.e., for $x\in W_n$
$$S(x)=\{y\in W_{n+1}: d(x,y)=1\}.$$

We denote by $\Phi=\{0,1,2,3\}$ the values
of the spins $\sigma(x)$ sitting on vertices
may assume.
A configuration on the Cayley tree is a collection $\sigma = \{\sigma(x),$ $x\in V\} \in \Phi^V$.

Consider a given subset ${\mathcal G}$
of pairs
$(i,j) \in \Phi \times \Phi$. Using ${\mathcal G}$ one can define a (partial) directed graph
on $\Phi$: if $(i,j)$ and $(j,i)$ are in ${\mathcal G}$ then the edge $(i,j)$ left undirected
(or equivalently have directions to both endpoints).
If $(i,j)\in {\mathcal G}$ but $(j,i)\notin {\mathcal G}$ then the edge $(i,j)$ is directed from $i$ to $j$.
The outdegree (respectly indegree) of a vertex is the number of ingoing (respectly outgoing) edges.
In this paper we assume that each vertex of the graph has positive indegree and outdegree.
The indegree will be denoted deg$^-(v)$ and the outdegree by deg$^+(v)$.

Note that ${\mathcal G}$ may be viewed as a directed graph  and that the correspondence is one-to-one (see Fig.~1).

\begin{center}
\includegraphics[width=8cm]{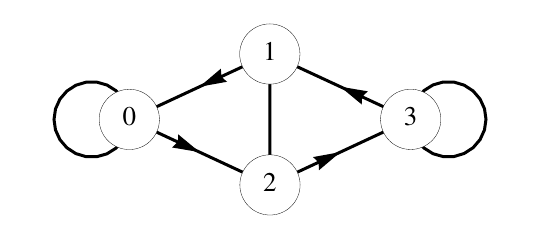}

{\footnotesize \noindent Fig.~1.
The directed graph associated with
$\mathcal G=\{(0,0),(0,2),(1,0),(1,2),(2,1),(2,3),(3,1),(3,3)\}$
(one puts an arrow when only one direction is selected).}
\end{center}

\begin{rk}
In this paper we consider the Cayley tree as a directed tree: an edge $\langle x,y\rangle\in L$ with endpoints $x$, $y$ has direction
from $x$ to $y$ iff $x<y$.  Thus the root $x^0$ has outdegree  $k+1$, each other vertex of the tree has indegree 1 (coming from the root)
and it has outdegree $k$.
\end{rk}

A configuration will be called admissible with respect to ${\mathcal G}$  if
 $(\sigma (x),\sigma (y)) \in {\mathcal G}$,
 for any pair of nearest neighbors $x$ and $y$ with $x<y$ (i.e., directed from $x$ to $y$).

 For a given set  ${\mathcal G}$ we denote by $\Omega$,
 the set of admissible configurations, and by $\Omega_A$ its restriction to a subset $A$ of $V$.

The  Hamiltonian
of the model is defined through a  matrix
$${\mathbf P}
=
%\bigg(P_{i,j}\bigg)_{i,j\in \Phi}=
\left(
\begin{array}{cccc}
P_{0,0}&P_{0,1}&P_{0,2}&P_{0,3}\\[2mm]
P_{1,0}&P_{1,1}&P_{1,2}&P_{1,3}\\[2mm]
P_{2,0}&P_{2,1}&P_{2,2}&P_{2,3}\\[2mm]
P_{3,0}&P_{3,1}&P_{3,2}&P_{3,3}\\[2mm]
\end{array}
\right), $$
where $P_{i,j}> 0$, if $(i,j)\in \mathcal G$; $P_{i,j}=0$ if $(i,j)\not\in \mathcal G$ and $\sum_{j\in \Phi}P_{i,j}=1$.

Namely, given ${\mathcal G}$ and ${\mathbf P}$,
we define the HC  Hamiltonian by
\begin{equation}\label{H}
H(\sigma)
%\equiv H^{\mathcal G}_{\mathbf P}(\sigma)
=\left\{\begin{array}{ll}
\sum_{\langle x, y\rangle}\log P_{\sigma(x),\sigma(y)},\ \ \mbox{if}\ \ \sigma\in \Omega,\\[4mm]
+\infty, \ \ \mbox{if}\ \ \sigma\not\in \Omega.\\
\end{array}\right.
\end{equation}

The paper is organized as follows.
Gibbs measures of the model with the corresponding system of  recursive equations are presented in Section 2.
For the theory of Gibbs measures on Cayley trees see \cite{R}.
In Section 3 we provide conditions under which the model has unique Gibbs measure.
Section 4 is devoted to  the diamond graph.
The results for fertile graphs are given in Section 5.

\section{Gibbs measures and recursive equations}

Let $t:\;x\in V\mapsto t_x=(t_{i,x}, \, i\in \Phi)\in{\R}^4_+$ be a vector-valued function on $V$.
Given $n=1,2,\ldots$,
consider the probability distribution $\mu^{(n)}$ on
$\Omega_{V_n}$ defined by
\begin{equation}\label{e5}
\mu^{(n)}(\sigma_n)=\frac{1}{Z_n}
\prod_{\langle x,y\rangle \subset V_n}
P_{\sigma(x),\sigma(y)}
\prod_{x\in W_{n}}t_{\sigma(x),x}.
\end{equation}
where the first product runs over pairs of nearest neighbors of $V_n$ and
$Z_n$ is the corresponding partition function.

We say that the probability distributions $\mu^{(n)}$
are compatible if $\forall$ $n\geq 1$ and
$\sigma_{n-1}\in\Omega_{V_{n-1}}$:
\begin{equation}\label{e7}
\sum_{\omega_n\in\Omega_{W_n}}
\mu^{(n)}(\sigma_{n-1},\omega_n)
=
\mu^{(n-1)}(\sigma_{n-1}).
\end{equation}

Such measures are usually called splitting Gibbs measures (see e.g. \cite{G,P,S}).

\begin{pro}\label{p4} The probability distributions
$\mu^{(n)}$, $n=1,2,\ldots$, in
(\ref{e5}) are compatible iff for any $x\in V$ the following
system of equations holds:
\begin{equation}\label{et}
z_{i,x}= \prod_{y\in S(x)}{P_{i,0}+P_{i,1}z_{1,y}+P_{i,2}z_{2,y}+P_{i,3}z_{3,y}\over
P_{0,0}+P_{0,1}z_{1,y}+P_{0,2}z_{2,y}+P_{0,3}z_{3,y}}.
\end{equation}
where $S(x)$ are the direct successors of $x$ (the $k$ nearest neighbors s.t. $x<y$) and   $z_{i,x}= t_{i,x}/t_{0,x}$, $i=1,2,3$.
\end{pro}

\begin{proof}  It consists to check condition (\ref{e7}) for the measures (\ref{e5}), see e.g.\ the proof of Theorem~1 in \cite{Ro15}.
\end{proof}

\section{Conditions of uniqueness }

\subsection{Condition on the graph ${\mathcal G}$}

As mentioned in the introduction,
every vertex of the graph is assumed to  have   positive indegree and outdegree (to avoid trivial situations).

\begin{rk}
Notice that if deg$^+(v)=1$ for all $v=0,1,2,3$, the corresponding hard core model has unique Gibbs measure.
\end{rk}

Indeed observe that in such a situation the  matrix ${\mathbf P}$ consists only 0 and 1, and  the Hamiltonian  (\ref{H})
reads
\begin{equation}\label{H1}
H(\sigma)=\left\{\begin{array}{ll}
0, \ \ \  \ \ \mbox{if}\ \ \sigma\in \Omega,\\[4mm]
+\infty, \ \ \mbox{if}\ \ \sigma\not\in \Omega.\\
\end{array}\right.
\end{equation}
In addition,  the set $\Omega$ of admissible configurations is finite,
so that
there exists a  unique Gibbs measure $\mu(\sigma)={1\over |\Omega|}$, $\sigma\in \Omega$.

\subsection{Condition on the matrix ${\mathbf P}$}

Denote $h_{i,x}=\ln z_{i,x}, i=1,2,3.$ Then the
equation (\ref{et}) can be written as

\begin{equation}\label{*1}
 h_{i,x}=\sum_{y\in S(x)}\ln{P_{i,0}+P_{i,1}\exp(h_{1,y})
 +P_{i,2}\exp(h_{2,y})+P_{i,3}\exp(h_{3,y})\over
P_{0,0}+P_{0,1}\exp(h_{1,y})+P_{0,2}\exp(h_{2,y})
+P_{0,3}\exp(h_{3,y})}.
\end{equation}

Note that $h_{i,x}\equiv 0,\, i=1,2,3,\, x\in V$ is a solution of (\ref{*1}).
Let us give  a condition on  $\mathbf P$ for
which it  will be the unique one.

We assume
\begin{equation}\label{*a}
P_{0,1}P_{0,2}P_{0,3}>0.
\end{equation}

\begin{lemma}
 If condition (\ref{*a}) is satisfied and
 $h_x=(\ln z_{1,x},\ln z_{2,x},\ln z_{3,x})$ is a solution
of (\ref{*1}) then
 $$z^-_i\leq z_{i,x} \leq z^+_i, $$
 for any
$i=1, 2, 3, \ \ x\in V$.
Here $(z^-_{1}, z^+_{1},  z^-_2, z^+_2, z^-_3, z^+_3)$ is a
solution of
\begin{equation}\label{*b}\begin{array}{llll}
\displaystyle z^-_i= \min_{(x,y,z)\in D}
f_i^k(x,y,z),
\\[2mm]
\displaystyle z^+_i=\max_{(x,y,z)\in D}f_i^k(x,y,z),\\
\end{array}
\end{equation}
where $D=[z_{1}^-,z_{1}^+]\times[z_2^-,z_2^+]\times [z_3^-,z_3^+]$
and
 $$f_i(x,y,z)={P_{i,0}+P_{i,1}x+P_{i,2}y+P_{i,3}z\over
P_{0,0}+P_{0,1}x+P_{0,2}y+P_{0,3}z}.$$
\end{lemma}
\begin{proof}  We rewrite (\ref{*1}) as
$$z_{i,x}= \prod^k_{j=1}f_i(z_{1,x_j}, z_{2,x_j}, z_{3,x_j}),$$
where
$x_j $
are  the direct successors of $x$.
The condition (\ref{*a}) guarantees that
the functions $f_i$ are bounded.
It is not difficult to see that
$$z^-_{i,1}<z_{i,x}<z^+_{i,1}, \ \ i=1, 2,3,$$
where
\begin{equation}\label{*2}\begin{array}{llll}
\displaystyle z^+_{i,1}=
\max_{{x, y, z>0}}f_i^k(x,y,z),\\[3mm]
\displaystyle z^-_{i,1}=\min_{{x, y, z>0}}f_i^k(x,y,z).
\end{array}
\end{equation}

Consider now the functions
$f_i(x,y,z)$ on  the sets $D_1=[z^-_{1,1}, z^+_{1,1}]\times[z^-_{2,1}, z^+_{2,1}]\times
[z^-_{3,1}, z^+_{3,1}]$.
A second step of the procedure leads to
$$z^-_{i,1}<z^-_{i,2}<z_{i,x}<z^+_{i,}<z^+_{i,1}, $$
and by iteration
we get the following inequalities
$$z^-_{i,n}<z_{i,x}<z_{i,n}^+,  $$
where  the $z^{\pm}_{i,n},
 n=1,2,...$,
 satisfy

\begin{equation}\label{e18}\begin{array}{llll}
\displaystyle z^-_{i,n+1}=\min_{(x,y,z)\in D_n}f^k_i(x,y,z),\\[3mm]
\displaystyle z^+_{i,n+1}=\max_{(x,y,z)\in D_n}f^k_i(x,y,z),\\
\end{array}
\end{equation}
with $z^{\pm}_{i,1}$ defined in (\ref{*2})
and
$$
D_n=[z^-_{1,n}, z^+_{1,n}]\times[z^-_{2,n}, z^+_{2,n}]\times
[z^-_{3,n}, z^+_{3,n}].$$

It is easy to see that  this  construction
leads to  bounded increasing (resp. decreasing)
sequences
$z^-_{i,n}$,   (resp.
$z^+_{i,n}$).
As a consequence,  we get the  existence of
$$\lim_{n\to\infty}z^{\pm}_{i,n}=z^{\pm}_i .
$$
This completes the proof.
\end{proof}

Consider the  function:
 $
h=(h_{1}, h_2, h_3)\to
F(h)=
(F_{1}(h),F_2(h),F_3(h))$,
 defined by
$$ F_i(h)
=\ln{P_{i,0}+P_{i,1}\exp(h_{1})+P_{i,2}\exp(h_{2})
+P_{i,3}\exp(h_{3})\over
P_{0,0}+P_{0,1}\exp(h_{1})+P_{0,2}\exp(h_{2})
+P_{0,3}\exp(h_{3})}.$$

We denote $\|h\|=\max \{|h_{1}|, |h_2|, |h_3|\}$ and  put
\begin{equation}
\theta_{ij}=\max_{(x,y,z)\in \mathcal D}\left|{\partial F_i(x,y,z)\over \partial h_j}\right|
\label{thetaij}
\end{equation}
where
$\mathcal D= [\ln z_{1}^-,\ln z_{1}^+]\times[\ln z_2^-,\ln z_2^+]\times [\ln z_3^-,\ln z_3^+].$

Condition (\ref{*a}) implies  that $\theta_{ij}<1$ for any $i,j=1,2,3$.
Let us denote
$$\theta=\max_{i,j}\theta_{ij}.$$
\begin{lemma}\label{l8} For any $h,l\in \mathcal D$ one has
\begin{itemize}
\item[a)] $\|F(h)-F(l)\|\leq 3\t\|h-l\|,$

\item[b)] $\|F(h)\|\leq 3\t \|h\|.$
\end{itemize}
\end{lemma}
\begin{proof} a)
We have
$$\|F(h)-F(l)\|=\max_{i=1,2,3}\{|F_i(h)-F_i(l)|\}\leq $$
$$\max_{i=1,2,3}\bigg\{\bigg|{\partial F_i\over \partial h_{1}}\bigg ||h_{1}-l_{1}|+\bigg|{\partial F_i\over \partial h_2}\bigg ||h_2-l_2|+
\bigg|{\partial F_i\over \partial h_3}\bigg ||h_3-l_3|\bigg\}\leq 3\t
\|h-l\|.$$

b) follows from a) taking into account that
$F(0,0,0)=0$  by letting $l=(0,0,0)$.
\end{proof}
\begin{thm}\label{t9} Under condition (\ref{*a}) and $$3k\t<1$$
the system of equations (\ref{*1})
has a unique solution
$h_{1,x}=h_{2,x}=h_{3,x}=0.$  Consequently there  exists a unique splitting Gibbs measure.
\end{thm}
\begin{proof}
 Using (\ref{*1}) and Lemma \ref{l8} we have
$$\|h_x\|\leq \sum_{y\in S(x)}\|F(h_y)\|\leq
k\max_{y\in S(x)}\|F(h_y)\| = (3k\t)\|h_{\tilde y}\|.$$ Iterating
this inequality leads to
\begin{equation}\label{e21}
 \|h_x\|\leq (3k\t)^n \|h_u\|,
 \end{equation}
where $u$ is  such that $d(x,u)=n$.
Since  $\|h_u\|\leq C=\max\{\ln z^+_{1}, \ln
z^+_2, \ln z^+_3\}$,  we get $h_x\equiv 0$.
This completes
the proof.
\end{proof}
\begin{rk} To check the condition of Theorem \ref{t9} one needs  a
solution of the system (\ref{*b}). But the analysis of solutions of
(\ref{*b}) is rather tricky.
However our construction gives a convenient way
to check the condition.
Namely,   one can check the condition
$3k\t^{(m)}<1$ where $\t^{(m)}=\max_{i,j}\{\t^{(m)}_{i,j}\}$.
Here $\t^{(m)}_{i,j}$ is defined by (\ref{thetaij})
with $z^{\pm}_i$ replaced by  $z^{\pm}_{i,m}$.
By construction  $\t^{(m)}_{i,j}\geq \t_{i,j}$
and $\lim_{m\to\infty}\t^{(m)}_{i,j}=\t_{i,j}.$ Thus $\t^{(m)}$ gives an
approximation for $\t.$
\end{rk}

\section{The diamond graph}
 Consider the  graph
 \begin{equation}\label{gi}
 \mathcal G_{\rm diamond}
 =\{(0,0),(0,2),(1,0),(1,2),(2,1),(2,3),(3,1),(3,3)\}.
 \end{equation}
 It is the graph shown in Fig.~1. It may be seen as compatibility rules on edges for a two state model.

Consider then the matrix
 \begin{equation}\label{ms}{\mathbf P}=
\left(
\begin{array}{cccc}
P_{0,0}=\alpha &P_{0,1}=0&P_{0,2}=1-\alpha &P_{0,3}=0
\\[2mm]
P_{1,0}=\beta &P_{1,1}=0 &P_{1,2}=1-\beta & P_{1,3}=0
\\[2mm]
P_{2,0}=0 &P_{2,1}=1-\beta & P_{2,2}=0&P_{2,3}=\beta
\\[2mm]
P_{3,0}=0&P_{3,1}=1-\alpha &P_{3,2}=0&P_{3,3}=\alpha
\\[2mm]
\end{array}
\right),
\end{equation}
where $\alpha, \beta\in (0,1)$.
This is a simplified version of the  diamond HC model with obvious symmetries between the parameters.

The corresponding set of recursive equations (\ref{et}) reads

\begin{equation}\label{et1}
f_x= \prod_{y\in S(x)}{\beta+(1-\beta)g_y\over
\alpha+(1-\alpha)g_y},\ \
g_x= \prod_{y\in S(x)}{\beta h_y+(1-\beta)f_y\over
\alpha+(1-\alpha)g_y}, \ \
h_x= \prod_{y\in S(x)}{\alpha h_y+(1-\alpha)f_y\over
\alpha+(1-\alpha)g_y},
\end{equation}
where $f_x=z'_{1,x}, \ \ g_x=z'_{2,x},\ \ h_x=z'_{3,x}$.

\subsection{Translation invariant measures}

\subsubsection{}

In this subsection we  look for solutions of the form
$f_x=f, \, g_x=g, \, h_x=h,$ for all    $x \in V$.

In this situation  we get from (\ref{et1}):

\begin{equation}\label{etu}
f= \left({\beta+(1-\beta)g\over
\alpha+(1-\alpha)g}\right)^k,\ \
g= \left({\beta h+(1-\beta)f\over
\alpha+(1-\alpha)g}\right)^k, \ \
h= \left({\alpha h+(1-\alpha)f\over
\alpha+(1-\alpha)g}\right)^k.
\end{equation}
Denoting $u=f^{1/k}$, $v=g^{1/k}$ and $w=h^{1/k}$ we obtain
\begin{equation}\label{ett1}
u= {\beta+(1-\beta)v^k\over
\alpha+(1-\alpha)v^k},\ \
v= {\beta w^k+(1-\beta)u^k\over
\alpha+(1-\alpha)v^k}, \ \
w= {\alpha w^k+(1-\alpha)u^k\over
\alpha+(1-\alpha)v^k}.
\end{equation}
The expression of $w$ as a function of $v$ reads
$$w=\left\{\beta^{-1}\left[v(\alpha+(1-\alpha)v^k)-(1-\beta)\left({\beta+(1-\beta)v^k\over
\alpha+(1-\alpha)v^k}\right)^k\right]\right\}^{1/k}.$$
Then we get

$$
v=\eta(v)\equiv {1\over \alpha+(1-\alpha)v^k}\left[(1-\beta)\left({\beta+(1-\beta)v^k\over
\alpha+(1-\alpha)v^k}\right)^k+\right.$$
\begin{equation}\label{er}
\left.\beta^{1-k}\left(\alpha v+(\beta-\alpha){(\beta+(1-\beta)v^k)^k\over
(\alpha+(1-\alpha)v^k)^{k+1}}\right)^k\right].
\end{equation}

\begin{lemma}\label{lt}
The function $\eta$ has the following properties:
\begin{itemize}
\item[1.] $\eta$ is a bounded function and $\eta(0)>0$, $\eta(+\infty)<+\infty$.
\item[2.] $\eta(1)=1$,\, $\eta'(1)=k\left[2\alpha-(1+k(\beta-\alpha))^2+k(\beta^2-\alpha^2)\right].$
\end{itemize}
\end{lemma}
\begin{proof} This  results from tedious but straightforward  computations.
\end{proof}

\begin{thm} If $\eta'(1)>1$, (i.e., $\alpha>{1\over k}$, $\alpha-{1\over k}<\beta<{k+1\over k-1}\left(\alpha-{1\over k}\right)$) then
there exist at least three translation-invariant Gibbs measures.
\end{thm}

\begin{proof} By Lemma \ref{lt},  $v=1$ is a solution of (\ref{er}).
When  $\eta'(1)>1$, $v=1$ is unstable.
So there exists a  small neighborhood $(1-\varepsilon, 1+\varepsilon)$ of $v=1$  such that   for $v\in (1-\varepsilon, 1)$ $\eta(v)<v$, and for $v\in (1, 1+\varepsilon)$
$\eta(v)>v$.
Since $\eta(0)>0$,  there exists a solution $v^*$
between $0$ and $1$. Similarly since $\eta(+\infty)<+\infty$ there is another solution $v^{**}$ between $1$ and $+\infty$. Thus, there exist at least three solutions. This completes the proof.
\end{proof}

The corresponding phase diagram is shown in Fig.~2.

\begin{minipage}[c]{9cm}
		\includegraphics[width=8cm]{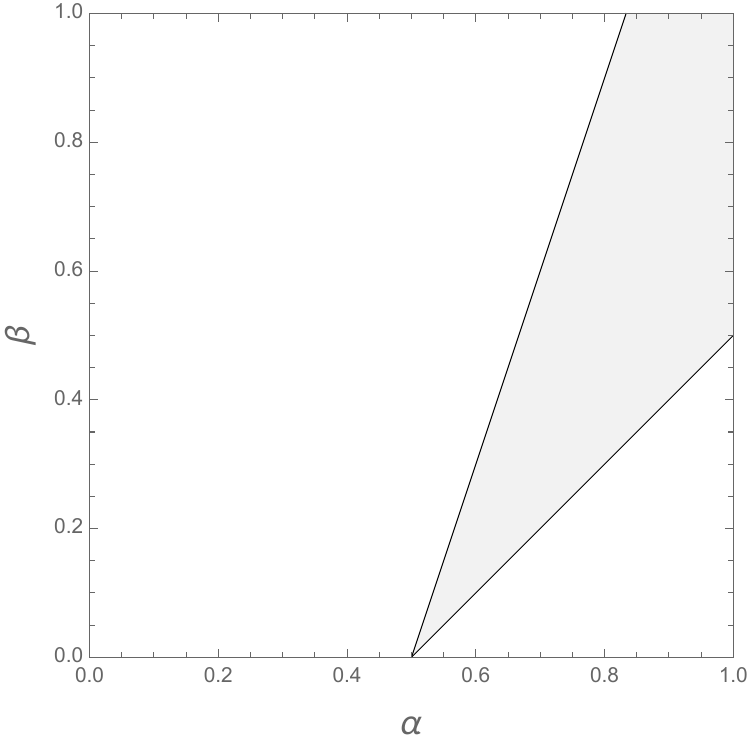}
\end{minipage}\hfill
\begin{minipage}[t]{9cm}
	      \includegraphics[width=3cm,height=3cm]{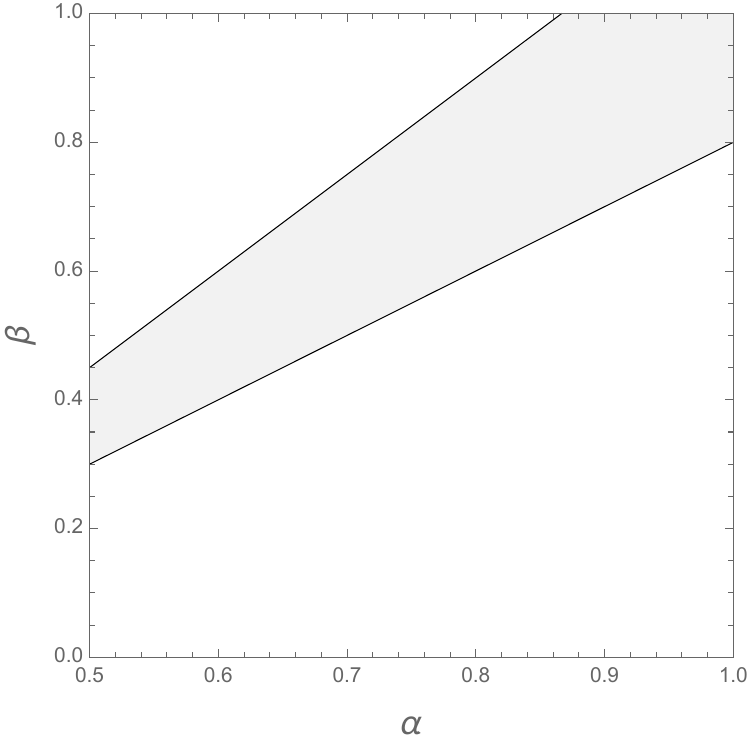}

\end{minipage}
{\footnotesize {\noindent Fig~2. Phase diagram  for $k=2$.
At least 3 solutions exist in the shaded area. Solution is unique in the white area.
Excerpt shows the case $k=5$.}}
\bigskip

\subsubsection{}
In this subsection we  look for solutions of the form
 $f_x=g_x$, $h_x=1$.
 Note that $f_x=g_x$, $h_x=1$ satisfies the system of equations (\ref{et1}) for any function $f_x$ which satisfies the following equation

\begin{equation}\label{et2}
f_x= \prod_{y\in S(x)}{\beta+(1-\beta)f_y\over
\alpha+(1-\alpha)f_y}.
\end{equation}
\begin{rk}
Recall that the functional equation of the Ising model on the Cayley tree is given by:

\begin{equation}\label{et3}
f_x= \prod_{y\in S(x)}{1+\theta f_y\over
\theta+f_y},
\end{equation}
where $\theta=e^{2J/T}$, $J$ denoting the strength of the interaction and $T$ the temperature.
The equations  (\ref{et2})  thus coincide with (\ref{et3} ) when
\begin{equation}
\alpha = {\theta \over \theta +1}
,
\quad
\beta= {1\over \theta +1}.
\end{equation}
Consequently all known results for Ising model can be reformulated for these particular values.
\end{rk}

To  give  a non-uniqueness condition for the  solutions of  the equation (\ref{et2}), we will use the following

\begin{lemma}\label{l1}
Let
\begin{equation}\label{ett}
z=\left({\beta+(1-\beta)z\over \alpha+(1-\alpha)z}\right)^k, \ \ z>0.
\end{equation}
Then,
\begin{itemize}
\item[1)] If $(\alpha,\beta)\in \left\{(x,y)\in [0,1]^2: x\leq {y(k+1)^3\over 4ky+(k-1)^2}\right\}$ the equation (\ref{ett}) has a unique solution $z=1$.

\item[2)] If $(\alpha,\beta)\in \left\{(x,y)\in [0,1]^2: x> {y(k+1)^2\over 4ky+(k-1)^2}\right\}$ the equation (\ref{ett}) has three solutions.
    \end{itemize}
\end{lemma}
\begin{proof} Denoting $x={1-\beta\over \beta}z$, $A={\beta(1-\alpha)^k\over (1-\beta)^{k+1}}$ and $B={\alpha(1-\beta)\over \beta(1-\alpha)}$ we get
\begin{equation}\label{epu}
Ax=\left({1+x\over B+x}\right)^k.
 \end{equation}
 This equation is studied in \cite{P}, Proposition 10.7. The equation (\ref{epu}) with $x\geq 0$, $k\geq 1$,
$A, B >0$ has a unique solution if either $k=1$ or $B\leq ({k+1
\over k-1})^2$. If $k>1$ and $B>({k+1 \over k-1})^2$ then there
exist $\nu_1(B,k)$, $\nu_2(B,k)$, with $0<\nu_1(B,k)< \nu_2(B,k)$,
such that the equation has three solutions if $\nu_1(B,k)<A<
\nu_2(B,k)$ and has two if either $A=\nu_1(B,k)$ or $A= \nu_2(B,k)$.
In fact:

\begin{equation}
\label{line}
\nu_i(B,k)={1\over x_i}\left({1+x_i \over B+x_i}\right)^k,
 \end{equation}
where $x_1,x_2$ are the solutions of
$$ x^2+[2-(B-1)(k-1)]x+B=0.$$
The critical line is obtained by inserting the critical value
  $B=\left({k+1\over k-1}\right)^2$ (equivalent to  $\alpha= {\beta(k+1)^2\over 4k\beta+(k-1)^2}$) in (\ref{line}).
\end{proof}

\begin{lemma}\label{l2}
The solutions  $f_x$  of (\ref{et2}) satisfy
$$ z^-\leq f_x\leq z^+,$$
where $z^-\leq 1\leq z^+$ solve the equation (\ref{ett}).
\end{lemma}
\begin{proof}
Let $\alpha>\beta$.

By using
properties of the function
$$\varphi(t)={\beta+(1-\beta)t\over
\alpha+(1-\alpha)t}$$
it is not difficult to see that
\begin{equation}
\label{z1}z^-_1
\equiv
\left({\beta\over \alpha}\right)^k
<f_x
<z^+_1
\equiv
\left({1-\beta\over 1-\alpha}\right)^k
\end{equation}

 Now
consider the function $\varphi$ on the interval
$[z^-_1, z^+_1]$.
A new iteration of the construction leads to
$$z^-_1<z^-_2<f_x<z^+_2<z^+_1.$$
The process of iterations give
$z^-_n<z_{i,x}<z_n^+,$
where $z^{\pm}_n, \ \ n=1,2,...$ satisfy
$$
z^-_{n+1}=\varphi^k(z^-_n),\ \
z^+_{n+1}=\varphi^k(z^+_n)
$$
It is easy to see that $z^-_n$ \ \ (resp.
$z^+_n$) are bounded increasing (resp. decreasing)
sequences.
This shows that the limits
$\lim_{n\to\infty}z^{\pm}_n=z^{\pm}$ exist.
 Moreover
$\varphi^k(z^{\pm})=z^{\pm}$.

(The case $\beta>\alpha$ is similar;
the case $\alpha=\beta$ is trivial).
\end{proof}

\begin{thm} \label{t1} \begin{itemize}
\item[1)] If $(\alpha,\beta)\in \left\{(x,y)\in [0,1]^2: x\leq {y(k+1)^2\over 4ky+(k-1)^2}\right\}$  the equation (\ref{et2}) has a unique solution $f_x\equiv 1$.
%Consequently, there exists  at least one translation-invariant Gibbs measure.

\item[2)] If $(\alpha,\beta)\in \left\{(x,y)\in [0,1]^2: x> {y(k+1)^2\over 4ky+(k-1)^2}\right\}$ the equation (\ref{et2}) has at least three solutions.
%Consequently, t

These solutions lead to  translation-invariant Gibbs measures.
\end{itemize}
\end{thm}

\begin{proof} 1)
It is easy to see that $f_x\equiv 1$ is a solution of (\ref{et2}).
From Lemma \ref{l1} it follows that
under  the conditions of theorem, the equation (\ref{ett}) has unique solution $z^-=z^+=1$. Then from Lemma \ref{l2}, one gets $f_x\equiv 1$.

2) It is a consequence of Lemma \ref{l1}. In this case,  equation (\ref{et2}) has at least three constant solutions  $f_x=1$, $f_x=z^+$ and $f_x=z^-$. \\
(see Fig.~3.)
\end{proof}

\begin{minipage}[c]{9cm}
\includegraphics[width=8cm]{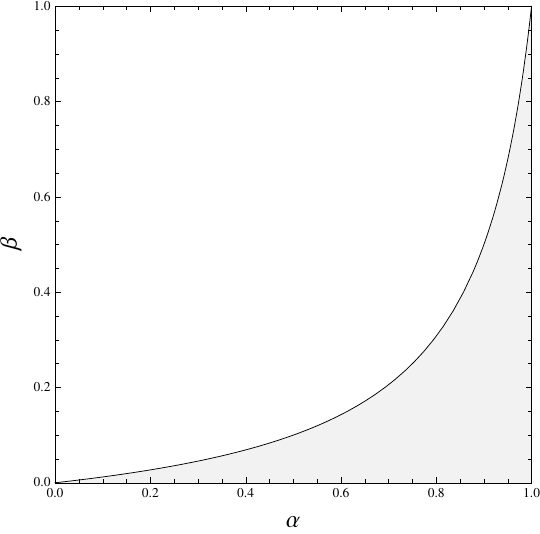}
\end{minipage}\hfill
\begin{minipage}[t]{9cm}
	         \includegraphics[width=3cm,height=3cm]{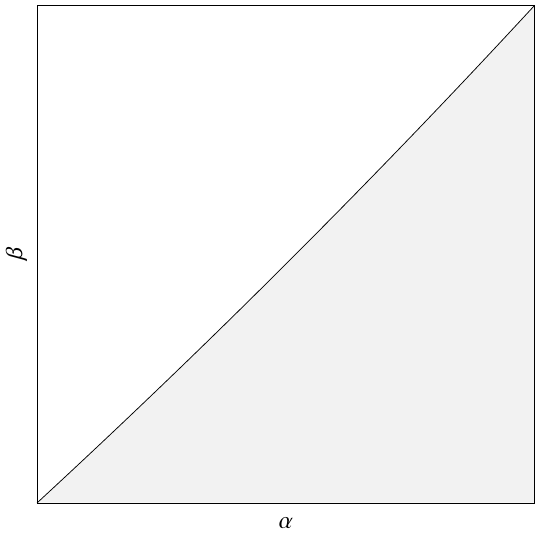}
\end{minipage}
{\footnotesize {\noindent Fig.~3. Case $f_x=g_x$, $h_x=1$ for $k=2$, equations (\ref{et1}).
At least 3 solutions exist in the shaded area. Solution is unique in the white area. Excerpt shows the case $k\gg 1$.}}
\bigskip

\subsection{Periodic measures}

In this subsection, we consider periodic solutions
of (\ref{et3}).
We will use the
group structure of the Cayley tree. It is known (see \cite{GR})
that there exists a one-to-one correspondence between the set of
vertices $V$ of a Cayley tree of order $k\geq 1$ and the group
$G_k$, free product of $k+1$ second-order cyclic groups with
generators $a_1, a_2, . . . , a_{k+1}$.

\begin{defn} Let ${\tilde G}$ be a normal subgroup of the group $G_k$. The set $z = \{z_x: x\in G_k\}$
 is said to be ${\tilde G}$-periodic if $z_{yx} =z_x$ for any $x\in G_k$ and $y\in {\tilde G}$.
 \end{defn}

\begin{defn} The Gibbs measure corresponding to a ${\tilde G}$-periodic set of quantities $z$ is said to be ${\tilde G}$-periodic.
\end{defn}

It is easy to see that a $G_k$-periodic measure is translation invariant.
Denote

$$G^{(2)}_k = \{x\in G_k: \, \mbox{the length of word} \, x \, \mbox{is even}\}.$$
This set is a normal subgroup of index two \cite{GR}. Note that $G^{(2)}_k$ is either the subset of even vertices (i.e. with even distance to the root).

The following proposition characterizes the set of all periodic solutions.

\begin{pro}\label{ty} For $\alpha\ne \beta$. Let $\tilde{G}$ be a normal subgroup of finite index in $G_k$.
Then each $\tilde{G}$- periodic solutions of equation (\ref{et2}) is either translation-invariant or $G^{(2)}_k$-
periodic.
\end{pro}
\begin{proof} It is easy to see that for $\alpha\ne \beta$ the function $\varphi(t)=(\beta+(1-\beta)t)/(\alpha+(1-\alpha)t)$ is one-to-one. Using this property together with arguments similar to the ones given in the proof of Theorem 2 in \cite{Mar1} lead to the statement.
\end{proof}

By Proposition \ref{ty}, the description of a $\tilde{G}$-periodic solutions of (\ref{et2}) is reduced to the solutions of
system (\ref{ep}) below. This system describes periodic solutions
with period two, more precisely, $G^{(2)}_k$-periodic solutions. They correspond to functions
$$f_x =\left\{\begin{array}{ll}
z_1, \ \ \mbox{if} \ \ x\in G_k^{(2)},\\
z_2, \ \ \mbox{if} \ \ x\in G_k \setminus G_k^{(2)}.
\end{array}\right.
$$
In this case, we have from (\ref{et2}):
\begin{equation}\label{ep}
z_1 =\left({\beta+(1-\beta)z_2\over \alpha+(1-\alpha)z_2}\right)^k, \ \ z_2 =\left({\beta+(1-\beta)z_1\over \alpha+(1-\alpha)z_1}\right)^k.
\end{equation}
Namely,  $z_1$ and $z_2$ satisfy
\begin{equation}\label{ef}
z = g(g(z)), \ \ \mbox{where}\ \ g(z) = \left({\beta+(1-\beta)z\over \alpha+(1-\alpha)z}\right)^k.
\end{equation}
Note that to get periodic (non translation invariant) measure we must find solutions
of (\ref{ep}) with $z_1\ne z_2$.
Obviously, such solutions are roots of the equation
\begin{equation}\label{ee} {g(g(z)) - z\over g(z) - z}= 0.
\end{equation}
For $k=2$, simple but long computations show that the last equation is equivalent to the equation
\begin{equation}\label{zz}
Az^2+Bz+C=0,
\end{equation}
where
$$A=[\alpha(1-\alpha)+(1-\beta)^2]^2, \ \  C=[\alpha^2+\beta(1-\beta)]^2,$$ $$B=4\alpha\beta(1-\alpha)(1-\beta)+2\beta(1-\beta)^3+\alpha^2(1-\beta)^2+2\alpha^3(1-\alpha)-
(1-\alpha)^2\beta^2.$$

The discriminant of the equation (\ref{zz}) has the following form
$$D=D(\alpha,\beta)=[3\alpha\beta(1-\alpha)(1-\beta)+\beta(1-\beta)^3+\alpha^3(1-\alpha)-
(1-\alpha)^2\beta^2]\times$$ $$
[5\alpha\beta(1-\alpha)(1-\beta)+3\beta(1-\beta)^3+2\alpha^2(1-\beta)^2+3\alpha^3(1-\alpha)-
(1-\alpha)^2\beta^2].$$
It is easy to see that $D(\alpha,\beta)=D(1-\beta, 1-\alpha)$. If $\alpha$ is small enough and $\beta$ is large enough then $D>0$ and $B<0$, in this case the equation $(\ref{zz})$ has two positive solutions.
Thus we have proved the following

\begin{thm}\label{tp1}
If $D>0$ and $B<0$ then the model corresponding to the matrix (\ref{ms}) has at least two $G^{(2)}_2$-periodic (non-translation-invariant) Gibbs measures
(see Fig.~4).
\end{thm}
\bigskip

\begin{center}
\includegraphics[width=8cm]{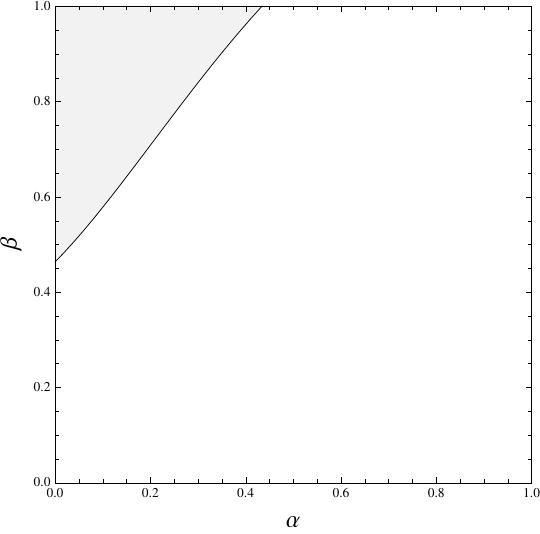}
\end{center}
{\footnotesize {\noindent Fig.~4. Case of periodic solutions of equations (\ref{et2}) for $k=2$. Two periodic solutions exist in the shaded area, no solution in the white area.}}
\bigskip

 In the next picture (Fig.~5), we collect the last three diagrams. One remarks that the transition curves of the cases $f_x = g_x, h_x = 1$ and $f_x =f; g_x =g; h_x =h$ are tangent to one another (this arises for all $k$).

\begin{center}

		\includegraphics[width=8cm]{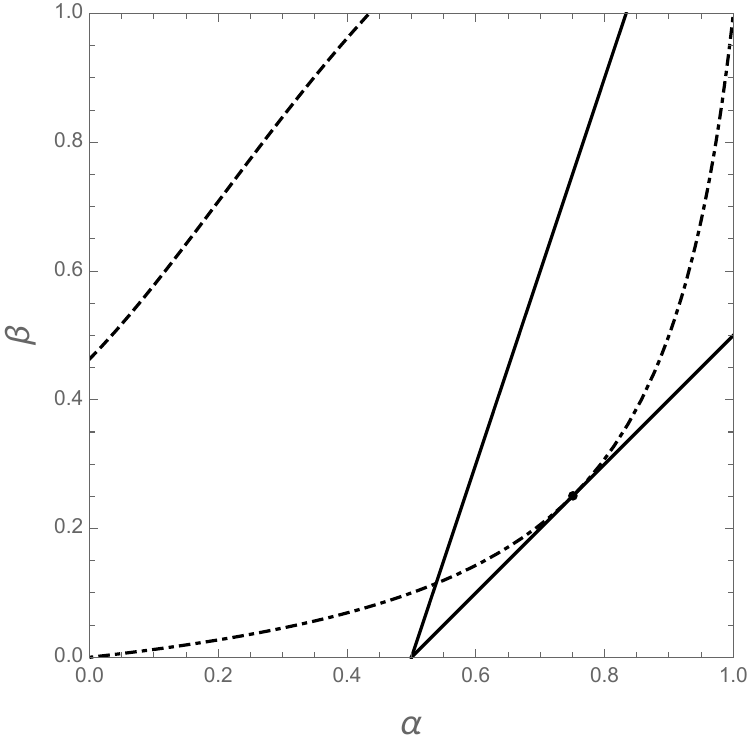}

{\footnotesize {\noindent Fig.~5. The 3 transition diagrams of Figs. 2, 3, and 4.}}

\end{center}
\bigskip

\section{Fertile graphs}

In this section, we consider symmetric graphs,
more precisely three types of fertile graphs,
the so-called stick, gun, and key \cite{Br1}:

\begin{eqnarray}
\mathcal G_{\rm stick}
  &=&
\{(0,1),(0,3),(2,3)\},
\nonumber
\\
\mathcal G_{\rm gun}
  &=&
\{(0,0),(0,1),(0,2),(0,3)(1,2)\},
  \nonumber
\\
\mathcal G_{\rm key}
&=&
\{(0,1),(0,2),(0,3),(1,2)\}.
\nonumber
\end{eqnarray}
 There, the above graphs are undirected, meaning that if $(a,b)$
belongs to the graph,  then it is also the case for  $(b,a)$.

\begin{rk}
The fertile graphs were defined in \cite{Br1}, they are the constraint graphs for which there
exists some set of activity and some Cayley tree for which the associated model exhibits multiple Gibbs measures.
Therefore, they are the interesting graphs to focus on in terms of identifying boundaries between
uniqueness/non-uniqueness regimes.  Brightwell and Winkler proved that a constraint graph is fertile
iff it has one of seven minimal fertile graphs as an induced subgraph.
The three minimal graphs on four vertices are exactly the stick, key and gun.
\end{rk}

\subsection{The  stick  graph}

For this graph, shown in Fig~6,
we have the matrix
 \begin{equation}\label{msa}{\mathbf P}=
\left(
\begin{array}{cccc}
P_{0,0}=0 &P_{0,1}=\alpha &P_{0,2}=0 &P_{0,3}=1-\alpha
\\[2mm]
P_{1,0}=1 &P_{1,1}=0 &P_{1,2}=0 & P_{1,3}=0
\\[2mm]
P_{2,0}=0 &P_{2,1}=0 & P_{2,2}=0&P_{2,3}=1
\\[2mm]
P_{3,0}=\beta&P_{3,1}=0 &P_{3,2}=1-\beta &P_{3,3}=0
\\[2mm]
\end{array}
\right),
\end{equation}
where $\alpha, \beta\in (0,1)$. Consequently,
the system of equations (\ref{et})
reads

\begin{equation}\label{st}\begin{array}{llllllllllll}
\displaystyle
f_x= \prod_{y\in S(x)}{1\over
\alpha f_y+(1-\alpha)h_y},\\[2mm]
\displaystyle
g_x= \prod_{y\in S(x)}{h_y\over
\alpha f_y+(1-\alpha)h_y}, \\[2mm]
\displaystyle
h_x= \prod_{y\in S(x)}{\beta+(1-\beta)g_y\over
\alpha f_y+(1-\alpha)h_y}.
\end{array}
\end{equation}

\begin{center}
\includegraphics[width=8cm]{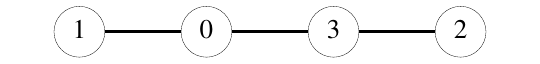}

{\footnotesize
\noindent Fig.~6. The stick  graph}
\end{center}

 Let us exhibit  conditions
on $\alpha$ and $\beta$ under which the system of equations (\ref{st})
has more than one constant solutions, i.e.
$f_x=f, g_x=g, h_x=h$.

We denote $u=f^{1/k}, v=g^{1/k}$ and $w=h^{1/k}$
we get from (\ref{st}) the following
\begin{equation}\label{st1}
u={1\over \alpha u^k+(1-\alpha)w^k}, \ \
v={w^k\over \alpha u^k+(1-\alpha)w^k}, \ \
w={\beta+(1-\beta)v^k\over \alpha u^k+(1-\alpha)w^k}.
\end{equation}

One easily finds that
$$u=\left(v(\beta+(1-\beta)v^k)^{-k}\right)^{1/(k+1)},
\ \ w=\left(v(\beta+(1-\beta)v^k)\right)^{1/(k+1)}.$$
Then from the second equation of (\ref{st1}) we get
\begin{equation}\label{st2}
v=Y(v)={1\over \alpha(\beta+(1-\beta)v^k)^{-k}+(1-\alpha)}.
\end{equation}

It is clear that $Y$ is an increasing,
bounded function and
$$
Y(0)={\beta^k\over \alpha+(1-\alpha)\beta^k}>0, \ \ Y(+\infty)={1\over 1-\alpha}<+\infty.$$
$$Y(1)=1, \ \ Y'(1)=k^2\alpha(1-\beta).$$

\begin{thm}
\label{tst} If $k^2\alpha(1-\beta)>1$ then there are at least three translation-invariant Gibbs measures.
\end{thm}
\begin{proof} By properties of $Y(v)$ we know that $v=1$ is a solution of (\ref{st2}).
Under $|Y'(1)|>1$, $v=1$ is unstable. So there is sufficiently small neighborhood of $v=1$: $(1-\varepsilon, 1+\varepsilon)$ such that $Y(v)<v$, for $v\in (1-\varepsilon, 1)$ and $Y(v)>v$, for $v\in (1, 1+\varepsilon)$. Since $Y(0)>0$ there is a solution $v^*$ between 0 and 1, similarly since $Y(+\infty)<+\infty$ there is an other solution $v^{**}$ between 1 and $+\infty$. Thus there are at least three solutions. This completes the proof.
\end{proof}

\subsection{The  gun graph}

For this graph (see Fig~7) one has
 \begin{equation}\label{msl}{\mathbf P}=
\left(
\begin{array}{cccc}
P_{0,0}=a &P_{0,1}=b &P_{0,2}=c &P_{0,3}=d
\\[2mm]
P_{1,0}=\alpha &P_{1,1}=0 &P_{1,2}=1-\alpha & P_{1,3}=0
\\[2mm]
P_{2,0}=\beta &P_{2,1}=1-\beta & P_{2,2}=0&P_{2,3}=0
\\[2mm]
P_{3,0}=1&P_{3,1}=0 &P_{3,2}=0 &P_{3,3}=0
\\[2mm]
\end{array}
\right),
\end{equation}
where $\alpha, \beta, a, b, c, d\in (0,1); a+b+c+d=1.$ Consequently,
the system of recursive equations (\ref{et})
is given by

\begin{equation}\label{gun}\begin{array}{llllllllllll}
\displaystyle
f_x= \prod_{y\in S(x)}{\alpha+(1-\alpha)g_y\over
a f_y+b g_y + c h_y+d},\\[2mm]
\displaystyle
g_x= \prod_{y\in S(x)}{\beta+(1-\beta)f_y\over
a f_y+b g_y + c h_y+d}, \\[2mm]
\displaystyle
h_x= \prod_{y\in S(x)}{1\over
a f_y+b g_y + c h_y+d}.
\end{array}
\end{equation}

\begin{center}
\includegraphics[width=7cm]{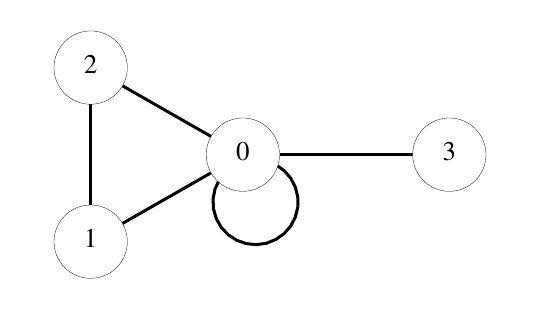}

{\footnotesize
\noindent Fig.~7. The gun  graph}
\end{center}

In this case for simplicity we assume $\alpha=\beta$.
Then for constant solutions, denoting $u=(f_x)^{1/k}, \,  v=(g_x)^{1/k}$ and $w=(h_x)^{1/k}$
we get form (\ref{gun}) that
\begin{equation}\label{gun1}
u={\alpha+(1-\alpha)v^k\over a u^k+b v^k+ c w^k+d}, \ \
v={\alpha+(1-\alpha)u^k\over a u^k+b v^k+ c w^k+d}, \ \
w={1\over a u^k+b v^k+ c w^k+d}.
\end{equation}
From this system we get
$$u=w(\alpha+(1-\alpha)v^k), \ \ v=w(\alpha+(1-\alpha)u^k).$$ Consequently
$$u(\alpha+(1-\alpha)u^k)=v(\alpha+(1-\alpha)v^k).$$
This gives $u=v$ and then $w=u(\alpha+(1-\alpha)u^k)^{-1}$.  Hence we have
\begin{equation}\label{gun2}
u=U(u)={\left(\alpha+(1-\alpha)u^k\right)^{k+1}\over \left[(a+b)\left(\alpha+(1-\alpha)u^k\right)^k+c\right]u^k+d\left(\alpha+(1-\alpha)u^k\right)^k}.
\end{equation}
The following properties of $U(u)$ are clear: $U$ is bounded and
$$U(0)={\alpha\over d}, \, U(+\infty)<+\infty, \, U(1)=1, \, U'(1)=k\{kc+d-\alpha(kc+1)\}.$$
Using these properties one can prove the following
\begin{thm}\label{tgun}  If $k(kc+d-\alpha(kc+1))>1$, there exists at least three translation-invariant Gibbs measures.
\end{thm}
\begin{proof} The proof is similar to the proof of Theorem \ref{tst}.
\end{proof}

\subsection{The  key graph}

The system of recursive equations for this  final fertile  graph under consideration (see Fig~8),  is as follows:

\begin{equation}\label{key}\begin{array}{llllllllllll}
\displaystyle
f_x= \prod_{y\in S(x)}{\alpha+(1-\alpha)g_y\over
a f_y+b g_y + c h_y},\\[2mm]
\displaystyle
g_x= \prod_{y\in S(x)}{\beta+(1-\beta)f_y\over
a f_y+b g_y + c h_y}, \\[2mm]
\displaystyle
h_x= \prod_{y\in S(x)}{1\over
a f_y+b g_y + c h_y},
\end{array}
\end{equation}
where $\alpha, \beta, a, b, c\in (0,1); a+b+c=1,$ which are defined as in (\ref{msl}) with $d=0$.

\begin{center}
\includegraphics[width=7cm]{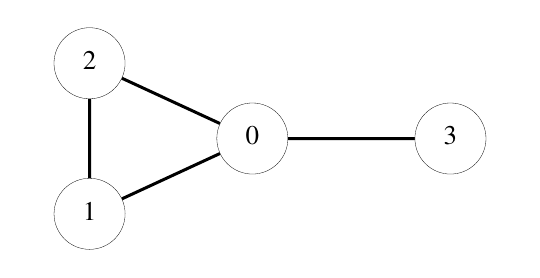}

{\footnotesize
\noindent Fig.~8. The key  graph}
\end{center}

This is  a particular case of the previously analyzed gun graph (obtained with $d=0$).
Hence Theorem \ref{tgun} remains true
with $\alpha=\beta$ and $k(kc-\alpha(kc+1))>1$.

\section*{ Acknowledgements}

 U.Rozikov thanks CNRS for support and  the Centre de Physique Th\'eorique De Marseille, France for kind hospitality
 during his several visits. We thank the referee for useful comments.

\end{document}